\documentclass[a4paper]{article}

	\usepackage[utf8]{inputenc}
	\usepackage[hyperref,standard,thmmarks]{ntheorem}
	\usepackage[pdftex, colorlinks=true, linkcolor = blue]{hyperref}
	\usepackage[all,cmtip]{xy}
	\usepackage{
    tikz,
    caption,
		a4,
		xspace,
		amsmath,
		amssymb,
		cleveref,
    xcolor,
    algorithm,algpseudocode,
	}
  \usepackage[textsize=footnotesize]{todonotes}
  \colorlet{red0}{black!75!red!100!}
  \colorlet{red1}{black!50!red!100!}
  \colorlet{red2}{black!25!red!100!}
  \colorlet{red3}{red}
  \colorlet{red4}{red!50!}
  \colorlet{red5}{red!25!}
  \colorlet{red6}{red!10!}
  \colorlet{red7}{red!5!}

	\renewtheorem{theorem}{Theorem}[section]
	\renewtheorem{proposition}[theorem]{Proposition}
	\renewtheorem{lemma}[theorem]{Lemma}
	\renewtheorem{corollary}[theorem]{Corollary}
	\theorembodyfont{\upshape}
	\renewtheorem{example}[theorem]{Example}
	\renewtheorem{remark}[theorem]{Remark}
	\renewtheorem{definition}[theorem]{Definition}

  \newcommand{\NN}{\mathbb N}

  \newcommand{\ZZ}{\mathbb Z}
  
  \newcommand{\RR}{\mathbb R}
  \newcommand{\DD}{\mathbb D}

  \newcommand{\B}{\mathcal B}

  \newcommand{\mto}{\rightrightarrows}

  \newcommand{\length}[1]{\left|#1\right|}
  \newcommand{\flength}[1]{\left|#1\right|}
  \newcommand{\sdzero}{\textup{\texttt{0}}}
  \newcommand{\sdone}{\textup{\texttt{1}}}
  \newcommand{\albe}{\Sigma}

  \newcommand{\bigo}{\mathcal O}
  \newcommand{\demph}[1]{\textbf{#1}}
  
  \newcommand{\eval}{\mathrm{eval}}

  \newcommand{\dom}{\mathrm{dom}}
  \newcommand{\xic}{\ensuremath{\xi_C}}
  \newcommand{\str}{\mathbf}
  \newcommand{\reg}{\albe^{**}}

  \newcommand{\abs}[1]{\left|#1\right|}

  \newcommand{\iRRAM}{\texttt{iRRAM}\xspace}

\title{A minimal representation for continuous functions}
\author{Franz Brauße\footnote{Universität Trier, 54286 Trier, room H~420; Email: brausse@informatik.uni-trier.de; supported by the German Research Foundation (DFG), project WERA, grant MU 1801/5-1}\ \ and Florian Steinberg\footnote{INRIA, Sophia-Antipolis; Email: florian.steinberg@inria.fr; Supported by the ANR project FastRelax(ANR-14-CE25-0018-01) of the French National Agency for Research}}
\date{}

\begin{document}

	\maketitle
	\abstract{
    Kawamura and Cook have specified the least set of information about a continuous function on the unit interval which is needed for fast function evaluation.
    This paper presents a variation of their result.
    To make the above statement precise, one has to specify what a \lq set of information\rq\ is and what \lq fast\rq\ should mean.
    Kawamura and Cook use polynomial-time computability in the sense of second-order complexity theory to define what \lq fast\rq\ means but do not use the most general \lq sets of information\rq\ this framework is able to handle.
    Instead they require codes to be length-monotone.
    This paper removes the additional premise of length-monotonicity, and instead imposes further conditions on the speed of the evaluation:
    The operation should now be computable in \lq hyper-linear\rq\ time.
    This means that the running time can not contain any iterations of the length function and, while an arbitrary polynomial may be applied to its value, on the argument side at most a shift by a constant is allowed.
    This is a very restrictive notion, but one can check that the Kawamura and Cook representation allows for hyper-linear time evaluation.
    The paper proves that it is not minimal with this property by providing the minimal set of information necessary for hyper-linear evaluation and proving that it is not polynomial-time equivalent to any encoding using only length-monotone names.
    This is ultimatively due to a failure of polynomial-time computability of an upper bound to a modulus of continuity.
  	Indeed this failure seems to reflect the behaviour of software based on the ideas of computable analysis appropriately and was one of the reasons for a closer investigation in the first place.
	}
	\tableofcontents
	\newpage
  \section{Introduction}
    This paper discusses subjects that are from the field of real complexity theory; The resource sensitive refinement of computable analysis.
    The goals of computable analysis and real complexity theory are to broaden the scope of classical computability and complexity theory from discrete structures to continuous structures.
    Computable analysis originates from one of the papers that is considered foundational for computability theory itself \cite{turing1936computable}.
    It branched of as a separate discipline in the 50's \cite{MR0086756} and has been extended steadily since.
    Nowadays, most researchers in computable analysis use Weihrauch's framework of representations \cite{MR1795407}.

    The complexity theory behind computable analysis was initiated by Friedman and Ko \cite{MR1137517} and has recently seen a lot of new developments due to advancements in the field of second-order complexity theory.
    Kawamura and Cook introduced a framework for complexity for operators in analysis \cite{MR2743298} and kicked off a line of investigations in the past years \cite[and many more]{kawamuraphd,MR3259646,MR3219039,MR3239272,MR3377508,postive,lmcs:3924}.
    One of the results that contributed to the popularity and acceptance of their framework is the following:
    Kawamura and Cook succeeded to provide a standard representation of the set of continuous functions on the unit interval.
    They proved that this representation contains the minimal information needed to make the evaluation operator polynomial-time computable.
    Where minimality is taken to mean that any other representation with this property can be translated to the standard representation in polynomial time.
    This paper provides a variation of Kawamura and Cook's result.

    The framework of Kawamura and Cook sits behind most complexity theoretical results in computable analysis.
    Still, there remains a gaps between the theory and applications:
    For a well-behaved complexity theory, Kawamura and Cook impose some additional assumptions on the representations they consider.
    In practice, these assumptions seem unnatural as they lead to extensive padding.
    Furthermore, some of the theoretical predictions seem to be out of sync with the behavior of efficient software based on the ideas from computable analysis:
    \iRRAM is a framework for and implementation of error-free real arithmetic based on the ideas of real complexity theory \cite{Mu00,iRRAM}.
    In \iRRAM it is possible to implement functions and, as long as the implementation of the function is reasonable, evaluation of the function is fast.
    Computing an upper bound of the modulus of continuity of a function, on the other hand, does not seem to be possible in a reasonable amount of time.
    In contrast to that, within Kawamura and Cook's framework one can prove that polynomial-time computability of evaluation implies polynomial-time computability of a modulus.

    Due to the additional assumptions Kawamura and Cook impose, namely length monotonicity of names, Kawamura and Cook only employ a fragment of second-order complexity theory.
    This paper asks the question whether the discrepancies between theory and practice in the specific application of representations of continuous functions on the unit interval can be removed by omitting length monotonicity.
    It should be pointed out ahead of time that while the approach seems to lead to a success in the beginning, we only consider it to be partially successful.
    Technical difficulties are encountered when composing functions.

    This paper provides a representation $\xic$ (\Cref{def:xic}) such that a function can be evaluated quickly by using an algorithm for evaluation that is very similar to how \iRRAM works internally (\Cref{res:hyper-linear time evaluation}).
    It is proven that it is impossible to compute an upper bound to the modulus of continuity of a function in polynomial-time with respect to $\xic$ (\Cref{resu:the modulus function}) and this is used to compare $\xic$ to Kawamura and Cook's minimal representation.
    While translatability in one direction follows from the minimality result proven by Kawamura and Cook, the representations are not polynomial-time equivalent (\Cref{resu:comparing representations}).
    It follows directly, that $\xic$ is not polynomial-time equivalent to any second-order representation (\Cref{resu:no second-order representation}).
		Many of the more basic operations, like the arithmetical operations, are polynomial-time computable with respect to the representation $\xic$.
		However, in contrast to Kawamura and Cook's representation, $\xic$ does not allow to extract an upper bound to the modulus of continuity in polynomial time.
		Furthermore, the final part of the paper proves composition of funcitons fails to be polynomial-time computable with respect to $\xic$ (\Cref{resu:composition}).

    The paper also proves that for any other representation such that evaluation is fast, there is a fast translation to $\xic$ (\Cref{resu:minimality of xic}).
    Here, the condition for being \lq fast\rq\ (\Cref{def:hyper-linear}) is more restrictive than polynomial-time computable and is given the name hyper-linear time computability.
    This notion leads to some technical difficulties.
    The use of a different notion of being \lq fast\rq\ is necessary for the proofs, but can also be justified by other means:
    In the past of real complexity theory there has been a lot of discussion about whether or not iteration of the length function in the running time should be considered feasible.
    Thus, one of the restrictions we use, namely forbidding iterations of the size function, is justifiable.
    The restriction, however, goes further to only allow a constant instead of the more usual polynomial lookahead.
    This seems to be a real restriction, and is only done since it seems unavoidable for the proofs.
    It should be noted that already the restriction to one iteration of the length function leads to a complexity class that is dependent on small changes in the model of computation.
    In the model that we pick, a consequence of this is that the class of operators that are considered \lq fast\rq\ is not closed under composition.

  \subsection{Notations}
    Fix the finite alphabet $\Sigma:=\{\sdzero,\sdone,\#\}$.
    Denote the set of finite words over $\Sigma$ by $\Sigma^*$.
    The empty string is denoted by $\varepsilon$.

    For convenience of notation, this paper considers some sets from mathematics as subsets of $\Sigma^*$:
    Let $\NN$ denote the set $\{\sdone,\sdone\sdzero,\sdone\sdone,\sdone\sdzero\sdzero,\sdone\sdzero\sdone,\ldots\}$ of \demph{positive integers in binary notation}.
    Let $\omega=\{\varepsilon,\sdone,\sdone\sdone,\ldots\}$ denote the \demph{non-negative integers in unary notation}.
    To avoid notational confusion this paper uses $2^n$ instead of $n$ if an integer in unary notation is handed to a machine.
    The length function $\length\cdot\colon \Sigma^*\to \omega$ assigns to a string $\str a$ its number of bits.
    Since $\omega$ are the integers in unary, this operation can also be regarded to replace all digits of the string by $\sdone$.
    Let $\ZZ$ denote the set $\sdzero\sdzero\NN\cup \sdzero\sdone\NN\cup\{\sdzero\sdzero\}$, where $\sdzero\sdzero$ is interpreted as $0$, $\sdzero\sdzero n$ is interpreted as $n$ and $\sdzero\sdone n$ is interpreted as $-n$.
    Finally, interpret a string $\str c$ that has a single $\#$ and starts in either $\sdzero\sdone$, $\sdone\sdone$ or $\sdzero\sdzero\#$ as the binary expansion of a rational number.
    I.e.\ identify $\str c$ with the rational number $(-1)^{c_0}(\sum_{i=1}^{m-1} c_i2^{i-(m-1)}+\sum_{i=m+1}^{\length{\str c}-1}c_i2^{i-m})$, where $m$ is the position of the $\#$.
    The set of numbers that have a code as above is called \demph{dyadic} numbers and denoted by $\DD$.
    Note that this does not provide $\DD\subseteq \Sigma^*$ but only defines partial a surjective mapping from $\Sigma^*$ to $\DD$, a so-called notation.
    Furthermore it holds that for any $n$ the $m+n$ initial segment of a dyadic number is again a dyadic number (where $m$ is the position of $\#$) and a $2^{-n}$-approximation to the original number.
    The above sets $\NN,\ZZ,\DD\subseteq \Sigma^*$ are pair-wise disjoint.

    The \demph{Baire space} $\B$ is the space of all string functions $\varphi:\Sigma^*\to\Sigma^*$.
    The reader is assumed to be familiar with the definitions of computability and complexity of string functions.
    The above can be used to talk about computability and complexity of functions between natural and dyadic numbers.
    Note that all string functions are required to be total, however, usually only the values of the functions on natural or rational inputs are required to fulfill some conditions.
    As a consequence it is possible to consider multivariate functions by just separating the arguments with $\#\#$.
    This paper uses the following pairing function on string functions:
    \[ \langle \varphi,\psi\rangle(\str a) := \begin{cases} \varphi(\str b) &\text{if }\str a= \sdzero\str b \\\psi(\str b) &\text{if }\str a =\sdone \str b \\ \varepsilon &\text{otherwise.}\end{cases} \]

    Throughout this paper $C([0,1])$ denotes the set of \demph{continuous real valued functions on the unit interval}.
    The following short notation for intervals is used:
    \[ [x\pm\epsilon]:= [x-\epsilon,x+\epsilon]. \]

  \subsection{Representations}

    Computability theory encodes discrete structures by strings.
    Since the set of all strings $\Sigma^*$ is countable, this can only work for countable structures.
    To compute on structures of continuum cardinality one has to encode the elements by string functions instead of strings.
    \begin{definition}\label{def:representation}
      A \demph{representation} $\xi$ of a space $X$ is a partial surjective mapping $\xi:\subseteq\B\to X$ from the Baire space to $X$.
    \end{definition}
    An element of $\xi^{-1}(x)$ is called a \demph{$\xi$-name} or simply a \demph{name} of $x$.
    An element of a space with a distinguished representation is called \demph{computable} resp.\ \demph{polynomial-time computable} if it has a name which is computable resp.\ polynomial-time computable.
    \begin{example}\label{ex:representation of reals}
      Throughout this paper, the real numbers are equipped with the following representation:
      A string function $\varphi$ is a name of $x\in\RR$ if and only if it holds for all $n\in\omega$ that
      \[ \varphi(2^n)\in\DD \text{ and }\abs{\varphi(2^n)-x}\leq 2^{-n}. \]
      That is: a name of a real number encodes dyadic approximations of arbitrary precision.
      This paper adopts the convention to encode precision requirements as integers in unary, which is standard in the field of real complexity theory.
      One could have equivalently used an integer in binary as input and replaced the right hand side by $\frac1{n+1}$ or a strictly positive rational $\epsilon$ that would then appear on the right hand side.
    \end{example}

    \begin{definition}\label{def:realizer}
      Let $\xi_X$ and $\xi_Y$ be representations of spaces $X$ and $Y$. A \demph{realizer} of a function $f\colon X\to Y$ is a function $F\colon\B\to\B$ such that for all $\varphi\in \B$
      \[ \varphi\in \dom(\xi_X) \quad\Rightarrow\quad \xi_Y(F(\varphi)) = f(\xi_X(\varphi)). \]
    \end{definition}
    That is: $F$ translates $\xi_X$-names of $x$ into $\xi_Y$-names of $f(x)$.
    Computability of operators on Baire space can be defined using oracle Turing machines:
    An operator $F:\subseteq\B\to\B$ is called \demph{computable} if there is an oracle Turing machine $M^?$ such that the run of $M^?$ on input $\str a$ and with oracle $\varphi\in\dom(F)$ halts with output $M^\varphi(\str a) = F(\varphi)(\str a)$.
    For more details about the exact model of oracle machines to use we point to \cite{MR2743298}.

    A function $f:X\to Y$ between spaces with distinguished representations is called \demph{computable} if it has a computable realizer.

    Finally, this paper needs the product construction.
    Recall that a pairing $\langle\cdot,\cdot\rangle$ of string functions was fixed in the introduction.
    \begin{definition}
      Let $\xi_X$ and $\xi_Y$ be representations of spaces $X$ and $Y$.
      Define a representation $\xi_{X\times Y}$ of the Cartesian product $X\times Y$ as follows:
      A string function $\varphi$ is a name of an element $(x,y)\in X\times Y$ if and only if there exist string functions $\psi\in\xi_X^{-1}(x)$ and $\psi'\in\xi_Y^{-1}(y)$ such that $\varphi=\langle \psi,\psi'\rangle$.
    \end{definition}
    Recall that an element of a represented spaces is called computable resp.\ polynomial time computable if it has such a name.
    It is true that an element $(x,y)$ of the product is computable resp.\ polynomial-time computable if and only if both $x$ and $y$ are computable resp.\ polynomial-time computable.

    \begin{example}
      For a given representation $\xi$ of the continuous functions on the unit interval $C([0,1])$, the above definitions together with the standard representation of the reals from \Cref{ex:representation of reals} allow to discuss computability and polynomial-time computability of the operator
      \[\tag{eval}\label{eq:eval} \eval:C([0,1])\times [0,1] \to \RR, \quad (f,x)\mapsto f(x). \]
    \end{example}

  \subsection{Second-order complexity theory}

    For complexity considerations this paper uses second-order complexity theory which goes back to a definition by Mehlhorn \cite{MR0411947}.
    However, just like the framework of Kawamura and Cook does, we replace the original definition by a characterization due to Kapron and Cook \cite{MR1374053}.
    This characterization is based on resource restricted oracle Turing machines and considerably more accessible than the original definition that was based on limited recursion on notation scheme.
    Recall that $\B:=\Sigma^*\to \Sigma^*$ denotes the Baire space, i.e. the space of all string functions.
    Oracle machines compute operators on Baire space and therefore take elements of Baire space as inputs.
    When bounding the running time of such a machine, the size of the functional input should be taken into consideration.
    \begin{definition}
      For a string function $\varphi\in\B$ define its \demph{length} $\flength\varphi:\omega\to\omega$ to be the function
      \[ \flength\varphi(n):= \max\{\length{\varphi(\str a)}\mid \length {\str a}\leq n\}. \]
    \end{definition}
		That is: the length of $\varphi$ is the worst case increase in string-size from input to output.
    A running time bound $T$ should be an object of the type $\omega^\omega\times \omega\to\omega$:
    It takes a size of an oracle function, a size of an input string and returns a number of steps $T(\flength{\varphi},\length{\str a})$ the machine is allowed to take on inputs $\varphi$ and $\str a$.
    The subclass of running times that are considered polynomial, i.e.\ the \demph{second-order polynomials}, are recursively defined as follows:
    \begin{itemize}
      \item Whenever $p$ is a polynomial with natural number coefficients, then the function $(l,n)\mapsto p(n)$ is a second-order polynomial.
      \item Whenever $P$ is a second-order polynomial, the function $(l,n)\mapsto l(P(l,n))$ is also a second-order polynomial.
      \item Whenever $P$ and $Q$ are second-order polynomials, then so are their point-wise sum and product.
    \end{itemize}
    \begin{definition}
      An oracle Turing machine $M^?$ is said to run in polynomial time on $A\subseteq \B$ if there is a second-order polynomial $P$ such that on oracle $\varphi\in A$ with input $\str a$ it halts after at most $P(\flength\varphi,\length{\str a})$ computation steps.
    \end{definition}
    A functional $F:\subseteq\B\to\B$ is called \demph{polynomial-time computable} if there is an oracle Turing machine $M^?$ that runs in polynomial time on $\dom(F)$ and such that for all $\varphi\in\dom(F)$ and strings $\str a$ it holds that $M^\varphi(\str a) = F(\varphi)(\str a)$.
    A function between spaces with distinguished representations is called \demph{polynomial-time computable} if it has a polynomial time computable realizer.

    It should be pointed out, that the characterization provided by Kapron and Cook only applies to the case where additional properties of the set $A$ are known.
    The definition stated here is a proper generalization in the sense that the operators we consider polynomial-time computable need not have polynomial-time computable total extensions.
    However, this seems to be a reasonable and necessary extension.

    An important special case where one is interested in computability or complexity of an operation are comparisons of different representations a space.
    \begin{definition}
      Let $\xi$ and $\xi'$ be representations of some space $X$.
      A \demph{translation} from $\xi$ to $\xi'$ is a realizer of the identity, i.e. a mapping $F:\subseteq \B\to\B$ such that for all $\varphi\in \B$ it holds that
      \[ \varphi \in\dom(\xi) \quad\Rightarrow\quad \xi'(F(\varphi))=\xi(\varphi). \]
      The representation $\xi$ is called \demph{topologically, computably} or \demph{polynomial-time translatable to} $\xi'$ if there exists a continuous, computable or poly\-no\-mi\-al-time computable translation.
      The representations $\xi$ and $\xi'$ are called \demph{topologically, computably} or \demph{polynomial-time equivalent} if there exist continuous, computable or polynomial-time computable translations in both directions.
    \end{definition}
    In literature the corresponding relation is usually called reducibility and denoted by $\preceq$.
    This terminology is taken from the discrete setting and can sometimes be confusing in the context of representations, as intuitively \lq $\xi$ is reducible to $\xi'$\rq\ should mean that there is a reduction mapping from $\xi'$ to $\xi$.

    \begin{example}
      The different versions of the representation of the real numbers discussed in \Cref{ex:representation of reals} lead to polynomial-time equivalent representations.
      Computability of functions is preserved under change to computably equivalent representations on both the input and output spaces.
      Polynomial-time computability is preserved under change of polynomial-time equivalent representations.
      These properties follow from the closure of computable and of polynomial-time computable operators under composition.
      A proof that the later remains true in our setting can for instance be found in \cite{kawamura_et_al:LIPIcs:2017:7737}.
    \end{example}

  \subsection{Hyper-linear time}

    Due to the use of general representations, this paper imposes the following more restrictive condition than polynomial-time computability on the evaluation operator:
    \begin{definition}\label{def:hyper-linear}
      A second-order polynomial $H$ is called \demph{hyper-linear}, if there exists some integer polynomial $p$ and a constant $C\in\omega$ such that
      \[ H(l,n) \leq p(l(n+C)+n) \]
      A polynomial-time computable function between represented spaces is called \demph{computable in hyper-linear time} if it is computed by a machine whose running time is bounded by a hyper-linear second-order polynomial.
    \end{definition}
    One should keep in mind that this definition is tailored for the application at hand.
    No care about complexity theoretical well-behavedness was taken.
    Indeed, the class of hyper-linear time computable operators may change with subtle changes in the model of computation.
    To make the above definition meaningful, more details about the model of computation have to be fixed:
    From now on assume that the position of the reading head resp. writing heads on the oracle tapes do not change during oracle queries and that oracle calls take one time step.

    \begin{example}
      Consider the two operators $F$ and $G$ defined by
      \[ F(\varphi)(\str a) := \varphi(\varphi(\str a))_1\quad\text{and}\quad G(\varphi)(\str a) := \overline{\varphi(\overline{\str a})}, \]
      where $\str a_i$ is the $i$-th bit of the string and $\overline{\str a} := \str a_{\length{\str a}}\ldots \str a_1
      $ is the mirrored string.
      The straight forward oracle machines that compute these operators run in time $\bigo(n+l(n))$.
      For $F$ this is due to our convention, that only reading the oracle tape is accounted for in the time consumption of the machine:
      While the return value might be very long, writing it to the output tape is done by the oracle and copying the first bit to the output tape takes constant time.
      Thus both $F$ and $G$ are hyper-linear-time computable.
      The composition $F\circ G$ of these operators is given by
      \[ (F\circ G)(\varphi)(\str a) = G(\varphi)(G(\varphi)(\str a))_1 = \overline{\varphi(\varphi(\overline {\str a}))}_1 = \varphi(\varphi(\overline{\str a}))_{\length{\varphi(\varphi(\str a))}}, \]
      and should intuitively not be hyper-linear-time computable.

      Indeed, it is not to difficult to give a proof that $F\circ G$ is not hyper-linear time computable:
      Assume that $M^?$ is a machine that computes $(F\circ G)$ in hyper-linear time $(l,n)\mapsto p(l(n+C)+n)$.
      Construct a pair of oracles $\psi_\sdzero$ and $\psi_\sdone$ such that $M^{\psi_\sdzero}(\varepsilon)=M^{\psi_\sdone}(\varepsilon)$ but $(F\circ G)(\psi_\sdzero)(\varepsilon)\neq(F\circ G)(\psi_\sdone)(\varepsilon)$.
      Let $\psi_i$ return the empty string on all arguments but $\varepsilon$, where it returns $\sdone^{C+1}$, and on the argument $\sdone^{C+1}$, where it returns $\sdone^{p(C+1)}i$:
      \[
				\psi_i (\str a):= \begin{cases}
				\sdone^{C+1} &\text{if } \str a=\varepsilon \\
				\sdone^{p(C+1)}i &\text{if } \str a=\sdone^{C+1}. \\
				\varepsilon & \text{otherwise}
			\end{cases} \]
      To see that $M^?(\varepsilon)$ returns identical results on both $\psi_i$ note that for all $n \leq C$ it holds that $\length {\psi_i} (n) = C+1$.
			Thus, the time the machine is granted of either of the oracles $\psi_i$ and input $\varepsilon$ is $p(|\psi|(0+C)+0)=p(C+1)$
      and the run does only rely on what is written in the first $p(C+1)$ cells of the oracle answer tape at any point in the computation.
      The content of this part of the oracle answer tape is identical for all possible answers of $\psi_\sdone$ and $\psi_\sdzero$.
      Thus the runs of the machine are identical and so is the return value.
      On the other hand, obviously $(F\circ G)(\psi_\sdzero)(\varepsilon) = \sdzero \neq \sdone = (F\circ G)(\psi_\sdone)(\varepsilon)$, thus the machine does not compute $F\circ G$.

      As the machine was arbitrary it follows that $F\circ G$ is not hyper-linear time computable.
    \end{example}
    This example shows that the hyper-linear-time computable operators are not closed under composition in the model of computation that we chose.
    The class is also not stable under rather minor changes in the model of computation.
    For instance, the alternate convention of counting one time step for each digit of the return value in an oracle query is fairly common throughout second-order complexity theory and leads to the same class of polynomial-time computable operators.
    We consider it to be less natural as it leads to doubled counting of steps when composing machines and more technical difficulties overall.
   	Making sense of hyper-linear time restrictions under this changed convention of time counting has to be done very carefully:
  	Whether or not a machine is allowed to abort an oracle query matters.
    If abort is disallowed, then being hyper-linear-time computable implies a polynomial lookahead which is too restrictive for the applications this paper is interested in.
    If aborting is allowed one has to ask again how this is done: aborting with an initial segment written to the answer tape leads to the same class of hyper-linear-time computable operators we work with.
    The convention where no information about the answer is available in case of an abort leads to again a different class not containing the operator $F$ from the previous example.

    All of the above difficulties equally apply to the class of machines that have a runtime bound of the form
    \[ (l,n)\mapsto p(l(q(n))+n). \]
    The class of operators computed by a machine allowing a running time bound of this form has been discussed as the right class for capturing feasibility in computable analysis.
		This justifies looking at hyper-linear time computation regardless of the model-dependence.
  \section{A minimal representation}

    Recall that this paper simulates multivariate input and output from $\NN$ or $\DD$ by separating the different arguments by $\#\#$
    and uses the abbreviation $[r\pm\epsilon]$ for $[r-\epsilon,r+\epsilon]$.
    This chapter proves the following representation to be the minimal representation such that evaluation is hyper-linear-time computable:
     \begin{definition}\label{def:xic}
      Define \demph{the representation $\xic$ of $C([0,1])$}:
      A string function $\varphi$ is a $\xic$-name of a function $f\in C([0,1])$ if and only if both of the following hold:
      \begin{enumerate}
        \item For all $r\in\DD\cap[0,1]$ and $n\in\omega$ there are $q\in\DD$ and $m\in\omega$ such that
        \[ \varphi(2^n\#\#r) = 2^m\#\#q \quad\text{and}\quad f([r\pm 2^{-m}]\cap[0,1])\subseteq [q\pm 2^{-n}]. \]
        \item For all $r,q\in\DD\cap[0,1]$ it holds that
        \[ \varphi(2^n\#\#r) = 2^m\#\#q\quad\Rightarrow\quad m\leq \flength{\varphi}(n). \]
      \end{enumerate}
    \end{definition}
    The first condition guarantees that on input $r$ and accuracy requirement $2^{n}$, a name of a function $f$ returns a $2^{-n}$-approximation $q$ of the value $f(r)$ of the function as well as an estimate $\delta:=2^{-m}$ of how much $r$ can be varied without the approximation $q$ becoming invalid.
   The second condition implies that $\flength\varphi$ is a modulus of continuity of $\xic(\varphi)$ in the following sense:
    A function $\mu:\omega\to\omega$ is called \demph{modulus of continuity} of $f\in C([0,1])$ if it fulfills
    \[\tag{mod}\label{eq:mod} \forall x,y\in[0,1]\colon \abs{x-y}\leq 2^{\mu(n)} \Rightarrow \abs{f(x)-f(y)}\leq2^{-n}. \]
    The above is automatically fulfilled for $\mu(n):=\flength\varphi(n+1)$ and $f:=\xic(\varphi)$.
    The length of a name can be increased arbitrarily without interfering with the other condition by changing the values of the string function on strings that do not contain any $\#$.
    Using this and the fact that any continuous function on the unit interval has a uniform modulus of continuity it is quite easy to see that the above indeed defines a representation, i.e. that any continuous function has a name.

    \begin{theorem}\label{res:hyper-linear time evaluation}
      The evaluation operator
      \[ \eval:C([0,1])\times[0,1]\to \RR, \quad(f,x)\mapsto f(x) \]
      is hyper-linear-time computable with respect to $\xic$.
    \end{theorem}
    \begin{proof}
      A machine computing the evaluation operator can be described as follows:
      When given a pair $\langle\varphi,\psi\rangle$ of a $\xic$-name $\varphi$ of a function $f\in C([0,1])$ and a name $\psi$ of a real number $x\in[0,1]$ and an precision requirement $2^{n}$ as input, the machine carries out the following loop for increasing $i$:
      First it obtains an encoding of a dyadic $2^{-i}$-approximation $x_i$ of $x$ by evaluating $\psi(2^{i})$.
      Then it evaluates $\varphi(2^n\#\#x_i)$ to obtain an encoding of a dyadic number $q_i$ and an integer $m_i$ such that $f([x_i\pm2^{m_i}])\subseteq [q_i\pm2^n]$.
      It checks if $m_i\leq i$.
      If this is not the case, it increases $i$ and restarts the loop.
      If it is the case it exits the loop and returns $q_i$.

      It should be clear that if the machine exits the loop at some point, then the return value is a valid approximation to $f(x)$.
      Therefore, it remains to prove that the machine always terminates and runs in polynomial time.
      Note that by the second condition of the definition of the representation $\xic$, the length of the name is a modulus of continuity.
      Claim that whenever $i\geq \flength\varphi(n)$, then the machine exits the loop.
      Indeed, in this case by the second condition of the definition of the representation $\xic$, it holds that $m_i \leq \flength\varphi(n)\leq i$.
      Thus, the loop is carried out at most $\flength\varphi(n)$ times.

      As the number $i$ is smaller than $\flength\varphi(n)$, going through the loop once takes hyper-linear time:
      The loop also needs to copy $2^n$, which takes $\bigo(n)$ steps.
      To see that copying the second argument $q_i$ of $\varphi(2^n\#\#x_i)$ is possible within the specified time bound, it is necessary to extract a bound on the integer part of $q_i$.
      This can be done as follows: The string $\sdzero\sdzero\#\sdone$ encodes the dyadic number $\frac12$.
      Thus, by the first condition of the definition of $\xic$ it holds that $\varphi(\sdone\#\#\sdzero\sdzero\#\sdone)=2^{m}\#\#q$ and $q$ and $m$ fulfill
      \[ f([1/2\pm2^{-m}]) \subseteq [q\pm1]. \]
      In addition to this, $\mu(n):=\flength\varphi(n+1)$ is a modulus of continuity of $f$ and by dividing the distance to any $x\in[0,1]$ to $\frac12$ into $2^{\flength\varphi(1)-1}$ steps of size less than $2^{-\flength\varphi(1)}$ it follows that
      \[ f([0,1])\subseteq [q\pm(1+2^{\flength\varphi(1)-1})]. \]
      This finally implies that the integer part of the second argument of the return value of $\varphi(2^n\#\#r)$ is smaller than $2^{\flength\varphi(1)}+2^{\flength\varphi(7)}$, where the second term is a bound on the integer part of $q$ that follows from how $q$ was found.
      Since $\flength\varphi(1)\leq\flength\varphi(7)$, such integers have codes that are of length less than $\flength\varphi(7)+3$.

      Therefore, the loop can be carried out in $\bigo(\max\{\flength\varphi(7),\flength\varphi(n),n\})\subseteq \bigo(n+\flength\varphi(n+7))$ steps and all of the computation takes less than $\bigo((n+\flength\varphi(n+7))^2)$.
      This time bound is hyper-linear.
    \end{proof}

  \subsection{A minimality property}
    With respect to the representation $\xic$ it is possible to evaluate in polynomial time.
    To prove that the representation is minimal with this property we need to provide a fast translation to $\xic$ for any other representation of the continuous functions on the unit interval that allows fast evaluation.
    \begin{theorem}\label{resu:minimality of xic}
      Let $\xi$ be a representation of $C([0,1])$.
      If the operator
      \[ \eval:C([0,1]) \times [0,1] \to \RR, \quad (f,x)\mapsto f(x) \]
      is hyper-linear-time computable with respect to $\xi$, then there exists a hyper-linear-time translation from $\xi$ to $\xic$.
    \end{theorem}
    \begin{proof}
      Assume the evaluation operator is computable in hyper-linear time.
      To build a machine that translates $\xi$ into $\xic$ proceed as follows:
      Given input of the form $2^n\#\#r$ (i.e. input for a $\xic$-name such that the first condition of \Cref{def:xic} applies) and a $\xi$-name $\varphi$ as oracle, execute a modified version of the source code of the evaluation operator on $2^n$:
      Note that the evaluation operator expects to be handed a pair $\langle\psi,\psi'\rangle$ of a $\xi$-name for the function and a name of a real number~$x$.
      Thus, whenever there is a leading $\sdzero$ on the query tape and a query command is issued, the machine first removes the leading $\sdzero$, and then queries the oracle.
      Whenever there is a leading $\sdone$ on the query tape, the oracle query command in the code of the evaluation are replaced with a code snippet that notes the maximum precision that was asked to the memory tape and then copies an appropriate initial segment of the encoding of the rational number $r$ to the oracle answer band.
      This produces an encoding of a dyadic number $q$ on the output tape.
      Finally the machine adds $2^m\#\#$ in front of the encoding, where $m$ is the highest precision that was required of the oracle for the real number and terminates.

      This produces a valid output of a $\xic$-name of $f$ on $2^n\#\#r$:
      The output is valid, as any $x\in[r\pm2^{-m}]$ has a name that returns the exact same initial segments of $r$ on queries less than $2^m$.
      The run of the evaluation operator on this oracle is identical to the run simulated above.
      Thus the return value is a valid approximation to $f(x)$ for each of these $x$.
      I.e. $f([r\pm2^{-m}])\subseteq [q\pm2^{-n}]$.

      To guarantee that the second condition from \Cref{def:xic} holds, recall that the evaluation operator being hyper-linear-time computable means that there is an integer polynomial $p$ and a natural number $C$ such that the run of the machine computing $\eval$ with oracle $\varphi$ on input $\str a$ takes at most $p(\flength\varphi(n+C)+n)$ steps.
      Let the machine proceed on inputs $\str a$ that are not of the form $2^n\#\#r$ as follows:
      For any of the $3^C$ strings $\str c$ of length $C$ it queries the oracle $\varphi$ on $\str c\str a$ and $\str c\str a'$, where $\str a'$ is the string where the first symbol after the first $\#$ is replaced by a $\#$ (and $\str a=\str a'$ if there is no $\#$ or the only one is the last symbol).
      It takes the maximum $m$ of the lengths of the oracle answers and returns the string consisting only of $\sdone$s and of length $p(m+n)$.

      The above guarantees that the string function produced by the machine has length bigger than $p(\flength\varphi(n+C)+n)$:
      Let $\str b$ be a string of length $n+C$ such that $\length{\varphi(\str b)} =\flength\varphi(n+C)$.
      Let $\str a$ be the last $n$ bits of $\str b$ where in the first occurrence of $\#\#$ the second $\#$ is replaced by $\sdzero$.
      Then the machine described above carries out the previous paragraph on input $\str a$.
      By the procedure described there it is guaranteed that the query $\str b$ is posed to the oracle and that the return value is longer than $p(\length{\varphi(\str b)}+n) = p(\flength\varphi(n+C)+n)$.

      The final thing to verify is that the second condition of the \Cref{def:xic} of $\xic$ is fulfilled by the function produced by the above procedure:
      Let $\psi$ be the string function produced by the machine above.
      By the previous it is clear that $\flength{\psi}(n)\geq p(\flength\varphi(n+C)+n)$.
      Since $(l,n)\mapsto p(l(n+C)+n)$ is a running time of the evaluation operator, which is simulated on an oracle of length $\flength{\varphi}$ and input $2^n$, it is clear that the number $m$ produced in the second paragraph of the proof is smaller than $p(\flength\varphi(n+C)+n)$ and therefore also as $\flength\psi(n)$.
    \end{proof}

    It should be noted, that the failure of closure under composition of hyper-linear-time computable operators has consequences for the applicability of the theorem.
    For instance, one would expect that the existence of a fast translation to the representation $\xic$ should imply that there exists an algorithm for fast evaluation.
    To obtain an algorithm for evaluation one has to first translate to $\xic$ and then use the algorithm for evaluation over $\xic$.
    As the class of hyper-liner time algorithms is not closed under composition, the algorithm obtained in this way need not run in hyper-linear time.
    It does run in polynomial time though.

   \subsection{Comparison to second-order representations}
    This chapter presents a hardness result for an operation with respect to the representation $\xic$:
    It is impossible to compute a modulus of continuity of a function in polynomial time with respect to $\xic$.
    This restriction is welcome as it seems to reflect the behavior of functions in \iRRAM.
    It should be noted that this result does not use the stronger notion of being \lq fast\rq\ that was previously used in this paper but really proves failure of polynomial-time computability.

    Computing a modulus of continuity is an inherently multivalued operation.
    Recall that a multivalued mapping $f:X\mto Y$ is an assignment of elements of $x$ to non-empty sets $f(x)\subseteq Y$.
    The elements of $f(x)$ are interpreted as the \lq acceptable return values\rq.
    \Cref{def:realizer} of a realizer can straight-forwardly be extended to apply to multivalued mappings and thus it makes sense to talk about computability and complexity of multivalued mappings.
    \begin{theorem}\label{resu:the modulus function}
      The modulus function
      \[ \mod:C([0,1])\mto \omega^\omega, f\mapsto \{\mu\mid \mu\text{ is mod. of cont. of } f\text{ (see \cref{eq:mod})}\} \]
      is not polynomial-time computable with respect to $\xic$.
    \end{theorem}
    \begin{proof}
      Towards a contradiction assume that there was a machine that computes a modulus of continuity in polynomial time.
      That is: There is a second-order polynomial $P$ such that the machine, when given a $\xic$-name $\varphi$ of a function $f$ and an input $2^n$ produces $2^{\mu(n)}$ on the output tape within $P(\length{\varphi},n)$ steps and the function $\mu$ is a modulus of continuity of $f$.
      Consider the following name $\psi$ of the constant zero function:
      \[ \psi(\str a) := \begin{cases} 2^{n}\#\#\sdzero\sdzero\# &\text{if }\str a= 2^{n+1}\#\# r\text{ for some }r\in\DD\cap[0,1] \\ \varepsilon & \text{otherwise.}\end{cases} \]
      Obviously $\flength{\psi}(n)=n+1$.
      The function $p(n):=P(\cdot+1,n)$ is a polynomial and bounds the number of steps until the machine returns some value $m$ of $\mu(n)$.
      Choose some $N$ such that $3p(N) < 2^N$.
      Consider the run of the machine on input $2^N$.
      Think of $[0,1]$ as the union of $2^N$ closed intervals of equal length $2^{-N}$.
      Since the $2^{-N-1}$ neighborhood of a rational number can at most intersect three such intervals, and the machine can at most ask $p(N)$ queries, at least one closed interval $I$ is such that no rational number in its $2^{-N-1}$ neighborhood is queried.
      Let $f'$ be the function that is zero everywhere but in $I$, where it takes the value $\frac32 2^{-N}$ in the middle and then goes linearly to zero with slope $3\cdot2^{\max\{\mu(N)-N,0\}}$.
      Note that any modulus of continuity of $f'$ at $N$ is strictly larger than $\max\{\mu(N),N\}$.

      To change the name $\psi$ of the zero function to a name $\psi'$ of $f'$ without changing any of the values the machine looked at during the computation, first note that due to the choice of the interval $I$ each query the machine makes is either a query with a precision such that zero is a valid approximation to the value of $f'$ or the name only returns information about the values on an interval disjoint from $I$.
      Therefore, it is possible to change the values of $\psi$ at strings the machine does not query to obtain a string function $\tilde \psi$ that fulfills the first condition of being a name of $f'$.
      Where the values the machine has not asked for can be chosen to be the exact values of $f'$ and the intervals can be chosen optimal.

      Furthermore, there are at least $2^M$ strings of length $M$ that do not represent any pair of a natural number and a dyadic number, for instance the binary strings.
      Thus, for any $M\geq N$ there is at least one such string $\str a_M$ the machine does not query.
      To obtain a valid name $\psi'$ of $f'$ change the values of $\tilde\psi$ on the string $\str a_M$ to have length according to a modulus of continuity of $f'$.

      As the machine behaves deterministically, and $\psi'$ and $\psi$ coincide on the values that are asked in the run with oracle $\psi$ and input $N$, the run of the machine on input $N$ with oracle $\psi'$ is identical and returns $\mu(N)$.
      However, by construction, $\mu(N)$ is not a value of any modulus of continuity of $f'$ in $N$.
      Therefore, no polynomial-time machine computing a modulus function exists.
    \end{proof}

    Kawamura and Cook introduced a framework for complexity considerations in analysis.
    For a well-behaved second-order complexity theory they impose an additional condition on the names:
    \begin{definition}[\cite{Kawamura:2012:CTO:2189778.2189780}]
      A string function $\varphi\in\B$ is called \demph{length-mo\-no\-to\-ne} if for all strings $\str a$ and $\str b$ it holds that
      \[ \length{\str a}\leq \length{\str b} \quad\Rightarrow\quad \length{\varphi(\str a)}\leq \length{\varphi(\str b)}. \]
      The set of all length-monotone string functions is denoted by $\reg$.
    \end{definition}
    The condition they impose is that any name in a representation is length-monotone.
    To distinguish their representations from the ones used in this paper we use their original terminology.
    \begin{definition}[\cite{Kawamura:2012:CTO:2189778.2189780}]
      A representation is a \demph{second-order representation} if its domain is contained in $\reg$.
    \end{definition}
    In this special case it is irrelevant whether time constraints are imposed on all of Baire-space or only for oracles from $\reg$.
    This may be attributed to the existence of a polynomial-time computable retraction from the Baire space to $\reg$ \cite{kawamura_et_al:LIPIcs:2017:7737} or verified directly.
     In particular, we may stick with the definition of polynomial-time computability used in the rest of this paper.
    \begin{definition}[\cite{Kawamura:2012:CTO:2189778.2189780}]
      Define a second-order representation $\delta_\square$ of $C([0,1])$ as follows: A length-monotone string function $\varphi$ is a name of a function $f\in C([0,1])$ if $\varphi=\langle\psi,\psi'\rangle$ for string functions $\psi$ and $\psi'$ that fulfill both of the following:
      \begin{enumerate}
        \item $n\mapsto \length{\psi(2^n)}$ is a modulus of continuity of $f$.
        \item for any encoding $r$ of a dyadic number in $[0,1]$ and $n\in\omega$ it holds that $\psi'(2^n\#\# r)$ is an encoding of a dyadic number $q$ and
        \[ \abs{f(r)-q}\leq 2^{-n}. \]
      \end{enumerate}
    \end{definition}
    A polynomial-time translation of $\delta_\square$ to $\xic$ is readily written down.
    The modulus function as defined in \Cref{resu:the modulus function} is obviously polynomial-time computable with respect to $\delta_\square$.
    With respect to $\xic$ the modulus function is not polynomial-time computable as proven in \Cref{resu:the modulus function}.
    Therefore, the representations $\delta_\square$ and $\xic$ are not polynomial-time equivalent.
    \begin{corollary}\label{resu:comparing representations}
      $\xic$ can not be translated to $\delta_\square$ in polynomial time.
    \end{corollary}
    Kawamura and Cook succeeded to prove the following:
    \begin{theorem}[Lemma 4.9 in \cite{Kawamura:2012:CTO:2189778.2189780}]
      For a second-order representation $\delta$ of $C([0,1])$ the following are equivalent
      \begin{itemize}
        \item The evaluation operator from \cref{eq:eval} is polynomial-time computable.
        \item $\delta$ is polynomial-time translatable to $\delta_\square$.
      \end{itemize}
    \end{theorem}
    Since the hyper-linear-time computability implies polynomial-time computability this entails the following:
    \begin{corollary}\label{resu:no second-order representation}
      $\xic$ is not polynomial-time equivalent to any second-order representation.
    \end{corollary}

  \subsection{Composition}

    This final chapter presents a major flaw of the representation $\xi_C$:
    It does not render the composition of functions polynomial-time computable.
    This makes it improbable that the representation $\xi_C$ is of value in applications.
    We believe that its study is of value nonetheless as its properties closely reflect well-known quirks of second-order complexity theory.
    It therefore outlines what can and cannot be done in real complexity theory when relying on second-order complexity theory.
    We like to believe that it provides evidence that one should either stick with the framework of Kawamura and Cook or go beyond the scope of second-order complexity theory.

    As a preparation note that an easy counting argument proves the following:
    \begin{theorem}\label{resu:polynomial-time and length}
      There does not exist any polynomial-time computable operator $F:\B\to\B$ such that
      \[ \forall\varphi\in\B.\forall n\in\omega\colon\length{F(\varphi)}(n) \geq \length{\varphi}(2n) \]
    \end{theorem}
    \begin{proof}
      Assume $M^?$ was a machine that computes an operator $F$ with the above property in time bounded by some second-order polynomial $P$.
      Consider the constant string function $\varphi(\str a)\equiv\varepsilon$.
      The length of this function is the constant zero function, thus $p(n):=P(\length{\varphi},n)$ is a polynomial.
      Since $P$ is a running time of $M^?$, the computation of $M^\varphi(\str  a)$ takes at most $p(\length{\str a})$ many steps for any input string $\str a$.
      Choose $N$ big enough such that $p(N)<2^N$.
      Note that there are $2^{2N}$ strings of length $2N$.
      The number of oracle queries $M^?$ asks for at least one input $\str a$ of length $N$ is bounded by $p(N)2^N<2^{2N}$.
      Thus, there exists at least one string $\str b$ of length $2N$ that is not queried during the computation of $M^\varphi(\str a)$ for any string $\str a$ of length less than $N$.
      Let $\psi$ be the function such that $\psi(\str b)=\sdzero^{\length{M^\varphi}(N)+1}$ and returns the empty string on all other values.
      The machine $M^?$ is deterministic and does not query $\str b$.
      Therefore it returns the same values with oracles $\varphi$ and $\psi$ and any input of length less or equal $N$.
      It follows that
      \[ \length{\psi}(2N)\geq \length{\psi(\str b)}=\length{M^\varphi}(N)+1 = \length{M^{\psi}}(N)+1 = \length{F(\psi)}(N)+1. \]
      This contradicts that the operator $F$ computed by $M^?$ has the desired property.
    \end{proof}
    This is in contrast to the situation in classical complexity theory, where for any polynomial $p\in\NN[X]$ there exists a polynomial-time computable function $\varphi$ such that $\length{\varphi(\str a)}\geq p(\length{\str a})$ for all input strings.
    The above proves that the straight forward translation of this statement to second-order complexity theory fails for the simplest second-order polynomials that are not hyper-linear.
    That the statement still holds true if the second-order polynomial is hyper-linear is what was made it possible to provide the minimality result for the representation $\xic$ from \Cref{resu:minimality of xic}.

    Also note that this theorem implies that there is no polynomial-time computable functional $F$ such that
    \[ \forall \varphi,\psi\in\B. \forall n\in\omega\colon \length{F(\langle \varphi, \psi\rangle)}(n) \geq (\length{\varphi}\circ\length{\psi})(n). \]
    As such an operator would provide an operator as in the theorem by fixing $\psi$ to be the function $\psi(\str a):=\str a\str a$.
    From this perspective it is not surprising that composition with respect to $\xic$ is not polynomial-time computable:
    Just like the failure of polynomial-time computability from \Cref{resu:the modulus function} lifted that the length function is not polynomial-time computable, the above can be lifted to infeasibility of composition.

    Let $C([0,1],[0,1])$ denote the set of all continuous functions whose image is contained in the unit interval.
    We consider this space a subspace of $C([0,1])$ and equip it with the range restriction of the representation $\xi_C$.
    The composition operator is defined as follows:
    \[ \circ\colon C([0,1])\times C([0,1],[0,1])\to C([0,1]),\quad (f,g)\mapsto f\circ g, \]
    where $(f\circ g)(x) := f(g(x))$.

    \begin{theorem}[Composition]\label{resu:composition}
      The composition operator is not polynomial-time computable with respect to the representation $\xi_C$.
    \end{theorem}

    \begin{proof}
      Towards a contradiction, assume that there exists a machine $M^?$ that runs in time bounded by a second-order polynomial $P$ and that when given a pair $\langle\varphi,\psi\rangle$ of $\xic$-names of functions $f\colon [0,1]\to \RR$ and $g\colon [0,1]\to[0,1]$ computes a $\xic$-name of $f\circ g$.

      Let $f$ be the following function:

      \noindent
      \begin{minipage}{.5\textwidth}
      \[ f(x) := \sum_{i=0}^\infty 2^{-i}\max\left\{1-\abs{2^{2i+2} x-3},0\right\}. \]
      Since $f$ is polynomial-time computable, it has a name $\varphi$ of polynomial length.
      Note that $f(0)=0$ and $f(\frac 34 2^{-2i}) = 2^{-i}$, in particular $f$ has no modulus smaller than $m\mapsto 2m$.
      \vspace{.12cm}
      \end{minipage}
      \hfill
      \begin{minipage}{.325\textwidth}
        \vspace{-.4cm}
        \begin{tikzpicture}
          \draw[->] (0,-.1) -- (0,3.2);
          \draw[->] (-.1,0) -- (4.2,0);
          \draw (-.1,3) -- (.1,3);
          \draw (4,-.1) -- (4,.1);
          \draw (2,-.1) -- (2,.1);
          \draw (1,-.1) -- (1,.1);
          \draw (.5,-.1) -- (.5,.1);
          \draw[blue] (4,0) -- (3,3) -- (2,0)--(1,0) -- (.75,1.5) -- (.5,0) -- (.25,0) -- (.125,0) -- (.09375,.75) -- (.0625,0) -- (0.03125,0) -- (0.0234375,0.3625) -- (0.015625,0) -- (0.0078125,0) -- (0.0078125,0.18125) ;
        \end{tikzpicture}
      \end{minipage}

      Consider the following name $\psi$ of the constant zero function $g$:
      \[ \psi(\str a) := \begin{cases} 2^{n}\#\#\sdzero\sdzero\# &\text{if }\str a= 2^{n+1}\#\# r\text{ for some }r\in\DD\cap[0,1] \\ \varepsilon & \text{otherwise.}\end{cases} \]
      Obviously $\flength{\psi}(n)=n+1$.
      The function $p(n):=P(\length{\varphi}+\length{\psi}+1,n)$ is a polynomial and bounds the number of steps until the machine returns some value.
      Choose some $N$ such that $3p(N) < 2^{N}$.

      Think of $[0,1]$ as the union of $2^{2N}$ closed intervals of equal length $2^{-2N}$.
      Since the $2^{2N-1}$ neighborhood of a rational number can at most intersect three such intervals, and the machine can at most ask $p(N)$ queries on each input $\str a$ of length $N$, there is at least one interval $I$ such that no query is asked in the $2^{2N-1}$ neighborhood of $I$.
      Let $g'$ be the function that is zero everywhere but in $I$, where it takes the value $\frac34 2^{-2N}$ in the middle and then goes linearly to zero with slope $3\cdot2^{\max\{p(N)-N,0\}}$.
      The argument that there is a valid name $\psi'$ of $g'$ such that the machine cannot distinguish it from $\psi$ can be copied from the proof of \Cref{resu:the modulus function}.

      Note that any modulus of continuity of $f\circ g'$ at $N$ is strictly larger than $\max\{p(N),N\}$ and that the runs of the machine on input $\str a$ of length less than $N$ are identical when the oracle $\langle\varphi,\psi\rangle$ is replaced by $\langle\varphi,\psi'\rangle$.
      Thus, the machine may not take more than $p(N)$ steps and can not produce a function whose length is a modulus of continuity of $f\circ g'$.

      This is a contradiction and thus no machine that computes the composition operator in polynomial time exists.
    \end{proof}

  \section{Conclusion}

    The representation $\xic$ was invented in an attempt to model the behavior of \iRRAM within the framework of second-order complexity theory.
    There is empirical evidence that within \iRRAM function evaluation is fast but computing a modulus of continuity is slow.
    The representation $\xic$ reflects this:
    It renders evaluation polynomial-time computable but does not allow to extract a modulus of continuity in polynomial time.
    It is remarkable that it is possible to do this within the framework of second-order complexity theory as previous results seemed to indicate that this is not possible.
    These very results forced us to leave the familiar setting of the framework for operators in analysis provided by Kawamura and Cook.

    However, the correspondence between $\xic$ and \iRRAM is imperfect:
    The running time of the straight forward algorithm for computing a modulus of continuity in \iRRAM is still way worse than that with respect to the representation $\xic$: Due to the possibility to brute force the length function, there is a cut of in the running time for functions with fast growing moduli that does not have an analogue in \iRRAM.
    It is improbable that this can be fully overcome as fast evaluation seems to necessitate the length to be comparable to a modulus of continuity.
    Furthermore, the representation $\xic$ has an undesirable property that is not reflected in the behavior of \iRRAM:
    Composition of functions is not polynomial-time computable with respect to $\xic$.

    In the proof of the hyper-linear-time computability of the evaluation operator with respect to $\xic$ in \Cref{res:hyper-linear time evaluation} the precision in each try is increased by one.
    This may lead to many useless queries.
    One could instead use the precision that the name requires the input approximation to have as next precision.
    However, this may lead to unnecessary high precision.
    Both approaches lead to comparable worst case complexities.
    The later, however, seems to be empirically superior as it is the approach that \iRRAM takes.

    \Cref{def:hyper-linear} of hyper-linear time could be slightly relaxed:
    The construction in \Cref{resu:minimality of xic} still works if the constant $C$ depends polynomially on the logarithm of $n$.
    If $C$ were allowed to depend on $n$ polynomially, the class would coincide with a class that some authors argue should be used to define polynomial-time computability anyway \cite{rett}.
    However, with respect to the convention of time consumption of oracle machines used in this paper, this bigger class is still not closed under composition.
    Furthermore, the technique used in \Cref{resu:minimality of xic} to prove the minimality of $\xic$ does not generalize.
    We think that it is unlikely that the proof can be recovered and believe that an argument similar to the one from the proof of the failure of the polynomial-time computability of the length function in \Cref{resu:polynomial-time and length} can be used to prove this.
    We did not attempt to carry this thought out as the rest of the paper is not concerned with this notion of polynomial-time computability.

  \bibliography{bib}{}

\end{document}